\documentclass[11pt]{article}%
\usepackage{amsfonts}
\usepackage{color}
\usepackage{amsmath}
\usepackage{fullpage}
\usepackage{amssymb}
\usepackage{graphicx}
\usepackage{hyperref}%
\providecommand{\U}[1]{\protect\rule{.1in}{.1in}}
\newtheorem{theorem}{Theorem}

\newtheorem{corollary}[theorem]{Corollary}

\newtheorem{lemma}[theorem]{Lemma}

\newtheorem{proposition}[theorem]{Proposition}

\newenvironment{proof}[1][Proof]{\noindent\textbf{#1.} }{\ \rule{0.5em}{0.5em}}
\numberwithin{equation}{section}

\let\originalleft\left
\let\originalright\right
\renewcommand{\left}{\mathopen{}\mathclose\bgroup\originalleft}
\renewcommand{\right}{\aftergroup\egroup\originalright}

\begin{document}

\title{Multiplicativity of completely bounded $p$-norms implies a strong converse for
entanglement-assisted capacity}
\author{Manish K. Gupta\\\textit{{\small {Hearne Institute for Theoretical Physics,}}} \\\textit{{\small {Department of Physics and Astronomy,}}} \\\textit{{\small {Louisiana State University,}}} \\\textit{{\small {Baton Rouge, Louisiana 70803, USA}}}
\and Mark M. Wilde\\\textit{{\small {Hearne Institute for Theoretical Physics,}}} \\\textit{{\small {Department of Physics and Astronomy,}}} \\\textit{{\small {Center for Computation and Technology,}}} \\\textit{{\small {Louisiana State University,}}} \\\textit{{\small {Baton Rouge, Louisiana 70803, USA}}} }
\date{\today}
\maketitle

\begin{abstract}
The fully quantum reverse Shannon theorem establishes the optimal rate of
noiseless classical communication required for simulating the action of many
instances of a noisy quantum channel on an arbitrary input state, while also
allowing for an arbitrary amount of shared entanglement of an arbitrary form.
Turning this theorem around establishes a strong converse for the
entanglement-assisted classical capacity of any quantum channel. This paper
proves the strong converse for entanglement-assisted capacity by a completely
different approach and identifies a bound on the strong converse exponent for
this task. Namely, we exploit the recent entanglement-assisted
\textquotedblleft meta-converse\textquotedblright\ theorem of Matthews and
Wehner, several properties of the recently established sandwiched R\'{e}nyi
relative entropy (also referred to as the quantum R\'{e}nyi divergence), and
the multiplicativity of completely bounded $p$-norms due to\ Devetak
\textit{et al}. The proof here demonstrates the extent to which the Arimoto
approach can be helpful in proving strong converse theorems, it provides an
operational relevance for the multiplicativity result of Devetak \textit{et
al}., and it adds to the growing body of evidence that the sandwiched
R\'{e}nyi relative entropy is the correct quantum generalization of the
classical concept for all $\alpha>1$.

\end{abstract}

\section{Introduction}

An important lesson learned in theoretical quantum information science is that
entanglement assistance tends to simplify problems of interest and, perhaps
surprisingly, makes such problems more like their classical counterparts. For
example, in quantum computational complexity theory, the canonical
QIP-complete problem is distinguishing two quantum channels specified by
quantum circuits \cite{RW05}. In order to solve this problem, two parties,
traditionally called a prover and verifier, engage in an entanglement-assisted
discrimination strategy, and furthermore, it is known that the computational
complexity of this task does not increase if these parties engage in a
completely classical strategy to solve this problem \cite{jain2010qip}. In the
theory of quantum error correction, the entanglement-assisted stabilizer
formalism allows for producing a quantum error-correcting code from an
arbitrary classical error-correcting code, while retaining the desirable
properties of the imported classical code \cite{BDH06,arx2006brun}. However,
importing arbitrary classical codes is not possible if entanglement assistance
is not available. Furthermore, entanglement assistance helps to resolve
technical problems that arise in the construction of quantum LDPC
\cite{HBD09,HYH09}\ and turbo codes \cite{WHB13}. In quantum rate distortion
theory (the theory of lossy quantum data compression), the most well
understood setting is again the entanglement-assisted setting, with there
being a simple formula that characterizes optimal compression rates
\cite{DHW11}.

Perhaps the earliest observation in this spirit is due to Bennett \textit{et
al}.~\cite{PhysRevLett.83.3081,Bennett2002} and Holevo \cite{Hol01a}, who
established that a simple formula characterizes the capacity of a quantum
channel for classical communication when unlimited entanglement assistance is
available. This result is one of the strongest in quantum Shannon theory and
provides a \textquotedblleft fully quantum\textquotedblright\ generalization
of Shannon's well known formula for the classical capacity of a classical
channel \cite{S48}. Furthermore, this formula is robust under the presence of
a noiseless quantum feedback channel from receiver to sender---Bowen
established that the entanglement-assisted capacity does not increase in the
presence of such a quantum feedback channel~\cite{B04}.

In later work, Bennett \textit{et al}.~\cite{BDHSW12} and Berta \textit{et
al}.~\cite{BCR09} strengthened the interpretation of the formula for
entanglement-assisted capacity, by demonstrating that a so-called strong
converse theorem holds in this setting. A strong converse theorem establishes
that, if the rate of communication in any given coding scheme exceeds the
capacity, then the error probability of this scheme tends to one in the limit
of many channel uses. Coupled with the achievability part of a coding theorem
(that there always exist a coding scheme with error probability tending to
zero in the limit of many channel uses if the rate of communication is less
than capacity), a strong converse theorem establishes the capacity as a very
sharp line dividing achievable communication rates from
unachievable ones. Furthermore, strong converse theorems find applications in
establishing security in particular models of cryptography \cite{KWW12}.

Bennett \textit{et al}.~and Berta \textit{et al}.~established a strong
converse theorem for entanglement-assisted capacity by proving what is known
as the entanglement-assisted quantum reverse Shannon theorem. Such a theorem
corresponds to a compression-like quantum information processing task and
characterizes the optimal rates of communication at which it is possible to
simulate a quantum channel. In entanglement-assisted channel simulation, the
goal is for a sender and receiver, who share an unlimited amount of
entanglement before the protocol begins, to use as few noiseless classical bit
channels as possible to simulate the action of many independent instances of
the channel on any quantum input (this input can be entangled with another
system not fed into the channel), in such a way that any third party should
not be able to distinguish between the original channels and the simulation.
Interestingly, the rate of communication required for channel simulation
corresponds to a strong converse rate for capacity, because, if it were
possible to use the channel to communicate at a rate larger than its channel
simulation rate, then a sender and receiver could \textquotedblleft get
out\textquotedblright\ more communication than they invested originally
(essentially getting \textquotedblleft something for nothing, bits for
free\textquotedblright). Carrying this \textit{reductio ad absurdum} argument
out in more detail, one can show that the error probability in fact increases
exponentially fast to one if the rate of communication exceeds the channel
simulation rate. The main result of Bennett \textit{et al}.~and Berta
\textit{et al}.~is that the optimal channel simulation rate in the presence of
shared entanglement is equal to the entanglement-assisted capacity of the
channel, and by the above argument, this establishes a strong converse theorem
for the entanglement-assisted capacity.

\section{Summary of results}

In this paper, we establish a strong converse theorem for the
entanglement-assisted capacity by a route completely different from that of
Bennett \textit{et al}.~and Berta \textit{et al}. Furthermore, we identify a
bound on the strong converse exponent for this task. Our motivation is
two-fold: first, for a theorem as important as this one, it is certainly
reasonable to have multiple proofs to shed further light on the topic. More
importantly, our approach here might illuminate alternate ways for
establishing a strong converse theorem for other capacities, the most pressing
of which is the quantum capacity of degradable channels \cite{MW13}. Our
approach taken here is in the line of Arimoto \cite{A73}, a route by which
several strong converse theorems have now been established
\cite{ON99,KW09,WWY13} and for which the general framework has been extended
significantly \cite{PV10,SW12}. Our proof makes use of the sandwiched
R\'{e}nyi relative entropy \cite{MDSFT13,WWY13}, known also by the name of
\textquotedblleft quantum R\'{e}nyi divergence\textquotedblright%
\ \cite{MDSFT13}, and it exploits several properties of this entropy in
order to establish the strong converse theorem.

A pleasing aspect of the present paper is that it provides an operational
relevance for the main result of Devetak, Junge, King, and Ruskai
\cite{DJKR06} (see related follow up work in \cite{J06}). Indeed, in the
present paper, for simplicity, we will \textit{define\footnote{Note that
\cite{DJKR06} defined the completely bounded norms $\left\Vert \cdot
\right\Vert _{\text{CB},p\rightarrow q}$ in a very general way and proved that
these norms reduce to the expression in (\ref{eq:CB-norm})\ when $p=1$ and
$q=\alpha$.}} the completely bounded $1\rightarrow\alpha$ norm of a completely
positive map $\mathcal{M}$ as%
\begin{equation}
\left\Vert \mathcal{M}\right\Vert _{\text{CB},1\rightarrow\alpha}\equiv
\max_{\rho_{A}}\left\Vert \left(  \rho_{A}^{1/2\alpha}\otimes I_{B}\right)
\Gamma_{AB}^{\mathcal{M}}\left(  \rho_{A}^{1/2\alpha}\otimes I_{B}\right)
\right\Vert _{\alpha}, \label{eq:CB-norm}%
\end{equation}
where $\Gamma_{AB}^{\mathcal{M}}$ is the Choi matrix of the channel
$\mathcal{M}$, $\left\Vert \cdot\right\Vert _{\alpha}$ denotes the Schatten
$\alpha$-norm for $\alpha\geq1$, and the optimization is over density
operators $\rho_{A}$ (see the next section for formal definitions of these
objects). A special case of the main result of \cite{DJKR06}\ is that these
norms are multiplicative for all completely positive maps $\mathcal{M}_{1}$
and $\mathcal{M}_{2}$, in the sense that%
\begin{equation}
\left\Vert \mathcal{M}_{1}\otimes\mathcal{M}_{2}\right\Vert _{\text{CB}%
,1\rightarrow\alpha}=\left\Vert \mathcal{M}_{1}\right\Vert _{\text{CB}%
,1\rightarrow\alpha}\left\Vert \mathcal{M}_{2}\right\Vert _{\text{CB}%
,1\rightarrow\alpha}.
\end{equation}
Due to the connection between the sandwiched R\'{e}nyi relative entropy and
$\alpha$-norms, we can apply the above multiplicativity result to our setting
in order to establish a strong converse for entanglement-assisted capacity.
Furthermore, it is interesting to observe that the completely bounded $\alpha
$-norm in (\ref{eq:CB-norm}) converges to the so-called \textquotedblleft
diamond norm\textquotedblright\ \cite{K97} in the limit as $\alpha
\rightarrow1$. The diamond norm is used all throughout quantum information
theory as a measure of distance between quantum
channels\ \cite{AKN98,GLN05,RW05,S05,W12}\ because it is the operationally
relevant distance measure in the setting of entanglement-assisted
discrimination of quantum channels (the most general strategy that one could
use to distinguish quantum channels). Of course, this is the setting with
which one would be dealing if trying to determine the value of a single bit
encoded with an entanglement-assisted communication strategy, so it appears
that our connection of the main result of \cite{DJKR06} to
entanglement-assisted capacity is the natural one to make.\footnote{The
authors of \cite{DJKR06} connected their main technical result to additivity
of a quantity now known as reverse coherent information \cite{GPLS09}.
However, in spite of the statements made in \cite{GPLS09}, we are not
convinced that the reverse coherent information possesses a compelling
operational interpretation.}

We now outline our proof for the strong converse of the entanglement-assisted
classical capacity, and the following sections give detailed arguments.

\begin{enumerate}
\item We say that a quantity is a generalized divergence \cite{PV10,SW12} (a
generalization of von Neumann relative entropy) if it satisfies the following
monotonicity inequality for all density operators $\rho$ and $\sigma$ and
channels $\mathcal{N}$:%
\begin{equation}
\mathbf{D}\left(  \rho\Vert\sigma\right)  \geq\mathbf{D}\left(  \mathcal{N}%
\left(  \rho\right)  \Vert\mathcal{N}\left(  \sigma\right)  \right)  .
\end{equation}
From such a divergence, we can derive a generalized mutual information of a
quantum channel according to the following recipe:%
\begin{align}
I_{\mathbf{D}}\left(  \mathcal{N}\right)   &  \equiv\max_{\rho_{A}%
}I_{\mathbf{D}}\left(  A;B\right)  _{\omega}\label{eq:recipe-first}\\
\omega_{AB} &  \equiv\rho_{A}^{1/2}\Gamma_{AB}^{\mathcal{N}}\rho_{A}^{1/2},\\
\Gamma_{AB}^{\mathcal{N}} &  \equiv\mathcal{N}_{A^{\prime}\rightarrow
B}\left(  \Gamma_{AA^{\prime}}\right)  ,\\
\Gamma_{AA^{\prime}} &  \equiv\left\vert \Gamma\right\rangle \left\langle
\Gamma\right\vert _{AA^{\prime}},\\
\left\vert \Gamma\right\rangle _{AA^{\prime}} &  \equiv\sum_{i}\left\vert
i\right\rangle _{A}\left\vert i\right\rangle _{A^{\prime}},\\
I_{\mathbf{D}}\left(  A;B\right)  _{\tau} &  \equiv\min_{\sigma_{B}}%
\mathbf{D}\left(  \tau_{AB}\Vert\tau_{A}\otimes\sigma_{B}\right)
.\label{eq:recipe-last}%
\end{align}
(We explain all of these quantities in further detail in the main text.)\ Our
first step then is to recall \cite[Propositions~20 and 21]{MW12}, which
establish a relationship between Type I and II errors in hypothesis testing
and the rate $R$ and success probability for any $\left(  n,R,\varepsilon
\right)  $ entanglement-assisted code (a code that uses $n$ instances of a
channel $\mathcal{N}$ at a fixed rate $R$ and has an error probability no
larger than~$\varepsilon$). We use this \textquotedblleft meta
converse\textquotedblright\ theorem in order to obtain an upper bound on any
entanglement-assisted code's success probability in terms of its rate,
blocklength, and any generalized mutual information derived from a generalized
divergence $\mathbf{D}\left(  \rho\Vert\sigma\right)  $ as above.

\item Next we recall the definition of the sandwiched R\'{e}nyi relative
entropy \cite{MDSFT13,WWY13}, also referred to as the quantum R\'{e}nyi
divergence \cite{MDSFT13}, which is a particular divergence between two
density operators $\rho$ and $\sigma$. It is defined \cite{MDSFT13,WWY13} for
$\alpha\in(1,\infty)$ as%
\begin{equation}
\widetilde{D}_{\alpha}\left(  \rho\Vert\sigma\right)  \equiv\frac{1}{\alpha
-1}\log\text{Tr}\left\{  \left(  \sigma^{\frac{1-\alpha}{2\alpha}}\rho
\sigma^{\frac{1-\alpha}{2\alpha}}\right)  ^{\alpha}\right\}  ,
\end{equation}
when the support of $\rho$ is contained in the support of $\sigma$ and it is
equal to $+\infty$ otherwise (in this work, we focus exclusively on the regime
$\alpha>1$). In particular, $\widetilde{D}_{\alpha}\left(  \rho\Vert
\sigma\right)  $ was shown to obey the monotonicity inequality mentioned above
for all $\alpha\in(1,2]$ \cite{WWY13,MDSFT13} and later work proved that this
monotonicity holds for all $\alpha\in\lbrack1/2,1)\cup(1,\infty)$
\cite{FL13,B13monotone}. Furthermore, this quantity converges to the von Neumann
relative entropy in the limit as $\alpha\searrow1$. This is one reason why the
sandwiched R\'{e}nyi relative entropy is relevant for us in establishing a
strong converse for entanglement-assisted capacity.\newline\newline%
\textit{Remark:} In a few recent works, it has been said, somewhat
ambiguously, that this new R\'{e}nyi relative entropy is useful because it
captures the \textquotedblleft non-commutativity of quantum
states.\textquotedblright\ However, the previous notion of R\'{e}nyi relative
entropy, $D_{\alpha}\left(  \rho\Vert\sigma\right)  \equiv\frac{1}{\alpha
-1}\log$Tr$\left\{  \rho^{\alpha}\sigma^{1-\alpha}\right\}  $ is perfectly
well defined for non-commutative quantum states and proves to be useful in the
regime when $\alpha\in\lbrack0,1)$
\cite{PhysRevLett.98.160501,PhysRevA.76.062301,N06,audenaert2008asymptotic,nussbaum2009chernoff,MH11,MO13}%
. In our opinion, the sandwiched R\'{e}nyi relative entropy has proved useful
when $\alpha>1$ because it can be related to a Schatten $\alpha$-norm in the
following way:%
\begin{equation}
\widetilde{D}_{\alpha}\left(  \rho\Vert\sigma\right)  =\frac{\alpha}{\alpha
-1}\log\left\Vert \sigma^{\frac{1-\alpha}{2\alpha}}\rho\sigma^{\frac{1-\alpha
}{2\alpha}}\right\Vert _{\alpha},
\end{equation}
and there are many properties of these $\alpha$-norms and prior results
established for them that come into play when establishing properties of
information measures derived from the sandwiched R\'{e}nyi relative entropy
(the present paper being no exception). Thus, there is a growing consensus
\cite{WWY13,MDSFT13,MO13,DL13} that the sandwiched R\'{e}nyi relative entropy
is the correct generalization of the classical R\'{e}nyi relative entropy at
least for the regime $\alpha>1$.

\item We evaluate the upper bound on success probability mentioned in the
first step above, by using the sandwiched R\'{e}nyi relative entropy as the
divergence. This yields the following upper bound on the success probability
of any rate $R$ entanglement-assisted scheme that uses a channel $n$ times:%
\begin{equation}
p_{\text{succ}}\leq2^{-n\sup_{\alpha>1}\left(  \frac{\alpha-1}{\alpha}\right)
\left(  R-\frac{1}{n}\widetilde{I}_{\alpha}\left(  \mathcal{N}^{\otimes
n}\right)  \right)  },
\end{equation}
where $\widetilde{I}_{\alpha}\left(  \mathcal{M}\right)  $ is the sandwiched
R\'{e}nyi mutual information of a quantum channel $\mathcal{M}$, derived by
the same recipe in (\ref{eq:recipe-first})-(\ref{eq:recipe-last}), taking
$\mathbf{D}=\widetilde{D}_{\alpha}$. By inspecting the above formula, we can
observe that if additivity of $\widetilde{I}_{\alpha}$ holds for $\alpha
\in(1,\infty)$, i.e.,%
\begin{equation}
\widetilde{I}_{\alpha}\left(  \mathcal{N}^{\otimes n}\right)  =n\widetilde
{I}_{\alpha}\left(  \mathcal{N}\right)  ,\label{eq:additivity-for-n}%
\end{equation}
then by a standard argument \cite{ON99,KW09}, which we elaborate for our case
here, the strong converse follows.

\item As a precursor to proving additivity of $\widetilde{I}_{\alpha}\left(
\mathcal{N}^{\otimes n}\right)  $, we relate the sandwiched R\'{e}nyi mutual information
of a channel $\mathcal{N}$ to an $\alpha$-norm of the states involved%
\begin{align}
\widetilde{I}_{\alpha}\left(  \mathcal{N}\right)   &  =\max_{\rho_{A}}%
\min_{\sigma_{B}}\widetilde{D}_{\alpha}\left(  \rho_{A}^{1/2}\Gamma
_{AB}^{\mathcal{N}}\rho_{A}^{1/2}\middle\Vert\rho_{A}\otimes\sigma_{B}\right)
\\
&  =\max_{\rho_{A}}\min_{\sigma_{B}}\frac{\alpha}{\alpha-1}\log\left\Vert
\left(  \left[  \Gamma_{AB}^{\mathcal{N}}\right]  ^{1/2}\left(  \rho
_{A}^{\frac{1}{\alpha}}\otimes\sigma_{B}^{\frac{1-\alpha}{\alpha}}\right)
\left[  \Gamma_{AB}^{\mathcal{N}}\right]  ^{1/2}\right)  \right\Vert _{\alpha}%
\end{align}
Using H\"{o}lder duality of norms and the Lieb concavity theorem \cite{L73},
we then show that%
\begin{equation}
\left\Vert \left(  \left[  \Gamma_{AB}^{\mathcal{N}}\right]  ^{1/2}\left(
\rho_{A}^{\frac{1}{\alpha}}\otimes\sigma_{B}^{\frac{1-\alpha}{\alpha}}\right)
\left[  \Gamma_{AB}^{\mathcal{N}}\right]  ^{1/2}\right)  \right\Vert _{\alpha}%
\end{equation}
is concave in $\rho_{A}$ for $\alpha\in(1,\infty)$. Convexity of the $\alpha
$-norm and operator convexity of $x^{\left(  1-\alpha\right)  /\alpha}$ for
$\alpha\in(1,\infty)$ implies that the above function is convex in $\sigma
_{B}$ for $\alpha\in(1,\infty)$. These properties are sufficient for us to
apply the Sion minimax theorem \cite{S58} in order to exchange the minimum
with the maximum for $\alpha\in(1,\infty)$:%
\begin{align}
\widetilde{I}_{\alpha}\left(  \mathcal{N}\right)   &  =\max_{\rho_{A}}%
\min_{\sigma_{B}}\widetilde{D}_{\alpha}\left(  \rho_{A}^{1/2}\Gamma
_{AB}^{\mathcal{N}}\rho_{A}^{1/2}\middle\Vert\rho_{A}\otimes\sigma_{B}\right)
\\
&  =\min_{\sigma_{B}}\max_{\rho_{A}}\widetilde{D}_{\alpha}\left(  \rho
_{A}^{1/2}\Gamma_{AB}^{\mathcal{N}}\rho_{A}^{1/2}\middle\Vert\rho_{A}%
\otimes\sigma_{B}\right)  .
\end{align}

\item From here, we can exploit the multiplicativity of completely bounded
$\alpha$-norms \cite{DJKR06} and \cite[Theorem~11]{B13monotone} to establish
that the sandwiched R\'{e}nyi mutual information is additive as a function of
quantum channels, in the sense that%
\begin{equation}
\widetilde{I}_{\alpha}\left(  \mathcal{N}_{1}\otimes\mathcal{N}_{2}\right)
=\widetilde{I}_{\alpha}\left(  \mathcal{N}_{1}\right)  +\widetilde{I}_{\alpha
}\left(  \mathcal{N}_{2}\right)  ,
\end{equation}
for all quantum channels $\mathcal{N}_{1}$ and $\mathcal{N}_{2}$ and all
$\alpha\in(1,\infty)$. This additivity result along with an inductive argument
gives us the additivity relation in (\ref{eq:additivity-for-n}).

\item Combining the above results, we obtain the following bound on the
success probability for any $(n,R,\varepsilon)$ entanglement-assisted coding
scheme for a channel $\mathcal{N}$:%
\begin{equation}
p_{\text{succ}}\leq2^{-n\sup_{\alpha>1}\left(  \frac{\alpha-1}{\alpha}\right)
\left(  R-\widetilde{I}_{\alpha}\left(  \mathcal{N}\right)  \right)  }.
\end{equation}
Finally, by a standard argument \cite{ON99,SW12} (which we elaborate for our
case here), we can choose $\varepsilon>0$ such that $\widetilde{I}_{\alpha
}(\mathcal{N})<{I}(\mathcal{N})+\varepsilon$ for all $\alpha>1$ in some
neighborhood of 1, so that the success probability decays exponentially fast
to zero with $n$ if $R>I(\mathcal{N})$, where $I(\mathcal{N})$ is the
entanglement-assisted capacity of the channel $\mathcal{N}$ (in this case,
$I(\mathcal{N})$ can be constructed according to the recipe in
(\ref{eq:recipe-first})-(\ref{eq:recipe-last}) with the generalized divergence
taken as the von Neumann relative entropy). The strong converse theorem for
the entanglement-assisted capacity then follows.
\end{enumerate}

The next section reviews some notations and definitions, and the rest of the
paper proceeds in the order above, giving detailed proofs for each step. We
then conclude with a brief summary.

\section{Notation and Definitions}

Let $\mathcal{B}\left(  \mathcal{H}\right)  $ denote the algebra of bounded
linear operators acting on a Hilbert space $\mathcal{H}$. We restrict
ourselves to finite-dimensional Hilbert spaces throughout this paper. The
$\alpha$-norm of an operator $X$ is defined as%
\begin{equation}
\left\Vert X\right\Vert _{\alpha}\equiv\text{Tr}\{(\sqrt{X^{\dag}X})^{\alpha
}\}^{1/\alpha},
\end{equation}
where $\alpha\geq1$. Let $\mathcal{B}\left(  \mathcal{H}\right)  _{+}$ denote
the subset of positive semi-definite operators (we often simply say that an
operator is \textquotedblleft positive\textquotedblright\ if it is positive
semi-definite). We also write $X\geq0$ if $X\in\mathcal{B}\left(
\mathcal{H}\right)  _{+}$. An\ operator $\rho$ is in the set $\mathcal{S}%
\left(  \mathcal{H}\right)  $\ of density operators if $\rho\in\mathcal{B}%
\left(  \mathcal{H}\right)  _{+}$ and Tr$\left\{  \rho\right\}  =1$. The
tensor product of two Hilbert spaces $\mathcal{H}_{A}$ and $\mathcal{H}_{B}$
is denoted by $\mathcal{H}_{A}\otimes\mathcal{H}_{B}$.\ Given a multipartite
density operator $\rho_{AB}\in\mathcal{S}\left(  \mathcal{H}_{A}%
\otimes\mathcal{H}_{B}\right)  $, we unambiguously write $\rho_{A}=\ $%
Tr$_{B}\left\{  \rho_{AB}\right\}  $ for the reduced density operator on
system $A$. A linear map $\mathcal{N}_{A\rightarrow B}:\mathcal{B}\left(
\mathcal{H}_{A}\right)  \rightarrow\mathcal{B}\left(  \mathcal{H}_{B}\right)
$\ is positive if $\mathcal{N}_{A\rightarrow B}\left(  \sigma_{A}\right)
\in\mathcal{B}\left(  \mathcal{H}_{B}\right)  _{+}$ whenever $\sigma_{A}%
\in\mathcal{B}\left(  \mathcal{H}_{A}\right)  _{+}$. Let id$_{A}$ denote the
identity map acting on a system $A$. A linear map $\mathcal{N}_{A\rightarrow
B}$ is completely positive if the map id$_{R}\otimes\mathcal{N}_{A\rightarrow
B}$ is positive for a reference system $R$ of arbitrary size. A linear map
$\mathcal{N}_{A\rightarrow B}$ is trace-preserving if Tr$\left\{
\mathcal{N}_{A\rightarrow B}\left(  \tau_{A}\right)  \right\}  =\ $Tr$\left\{
\tau_{A}\right\}  $ for all input operators $\tau_{A}\in\mathcal{B}\left(
\mathcal{H}_{A}\right)  $. If a linear map is completely positive and
trace-preserving (CPTP), we say that it is a quantum channel or quantum
operation. A positive operator-valued measure (POVM) is a set $\left\{
\Lambda^{m}\right\}  $ of positive operators such that $\sum_{m}\Lambda^{m}=I$.

The sandwiched R\'{e}nyi relative entropy \cite{MDSFT13,WWY13}, also referred
to as the quantum R\'{e}nyi divergence \cite{MDSFT13}, between two density
operators $\rho$ and $\sigma$ is defined for $\alpha\in(1,\infty)$\ as
follows:%
\begin{equation}
\widetilde{D}_{\alpha}\left(  \rho\Vert\sigma\right)  \equiv\frac{1}{\alpha
-1}\log\text{Tr}\left\{  \left(  \sigma^{\frac{1-\alpha}{2\alpha}}\rho
\sigma^{\frac{1-\alpha}{2\alpha}}\right)  ^{\alpha}\right\}  ,
\end{equation}
whenever the support of $\rho$ is contained in the support of $\sigma$ and it
is equal to $+\infty$ otherwise. Throughout this work, we will be considering
only the range $\alpha\in(1,\infty)$. For such choices, the H\"{o}lder
conjugate of $\alpha$ is $\alpha^{\prime}$ such that $\frac{1}{\alpha}%
+\frac{1}{\alpha^{\prime}}=1$, so that $\alpha^{\prime}=\alpha/\left(
\alpha-1\right)  \in(1,\infty)$. We can define a sandwiched R\'{e}nyi mutual
information of a bipartite state $\rho_{AB}$\ as%
\begin{equation}
\widetilde{I}_{\alpha}\left(  A;B\right)  _{\rho}\equiv\min_{\sigma_{B}%
}\widetilde{D}_{\alpha}\left(  \rho_{AB}\Vert\rho_{A}\otimes\sigma_{B}\right)
.
\end{equation}
Let $\left\vert \Gamma\right\rangle _{AA^{\prime}}$ denote the
\textquotedblleft maximally-entangled-like\textquotedblright\ vector:%
\begin{equation}
\left\vert \Gamma\right\rangle _{AA^{\prime}}\equiv\sum_{i}\left\vert
i\right\rangle _{A}\left\vert i\right\rangle _{A^{\prime}}.
\end{equation}
We can then define the sandwiched R\'{e}nyi mutual information of a channel as
\begin{equation}
\widetilde{I}_{\alpha}\left(  \mathcal{N}\right)  \equiv\max_{\rho_{A}%
}\widetilde{I}_{\alpha}\left(  A;B\right)  _{\omega},
\end{equation}
where%
\begin{equation}
\omega_{AB}\equiv\rho_{A}^{1/2}\mathcal{N}_{A^{\prime}\rightarrow B}\left(
\Gamma_{AA^{\prime}}\right)  \rho_{A}^{1/2}.
\end{equation}
In what follows, we will use the following abbreviation for the Choi matrix:
\begin{equation}
\Gamma_{AB}^{\mathcal{N}}\equiv\mathcal{N}_{A^{\prime}\rightarrow B}\left(
\Gamma_{AA^{\prime}}\right)  ,
\end{equation}
so that%
\begin{equation}
\omega_{AB}=\rho_{A}^{1/2}\Gamma_{AB}^{\mathcal{N}}\rho_{A}^{1/2}.
\end{equation}

The quantum relative entropy $D\left(  \rho\Vert\sigma\right)  $ is defined as%
\begin{equation}
D\left(  \rho\Vert\sigma\right)  \equiv\text{Tr}\left\{  \rho\left[  \log
\rho-\log\sigma\right]  \right\}  ,
\end{equation}
whenever the support of $\rho$ is contained in the support of $\sigma$ and it
is equal to $+\infty$ otherwise.

\section{Bounding the success probability of any entanglement-assisted code
with a generalized divergence}

We first review the steps in any\ general $\left(  n,R,\varepsilon\right)  $
protocol for entanglement-assisted classical communication over $n$ uses of a
quantum channel. Such a protocol begins with a sender Alice and a receiver Bob
sharing an arbitrary bipartite entangled state $\Psi_{T_{A}T_{B}}$, where
Alice possesses the system $T_{A}$ and Bob the system $T_{B}$. Their goal is
to use the entangled state $\Psi_{T_{A}T_{B}}$ and $n$ instances of a noisy
channel $\mathcal{N}_{A^{\prime}\rightarrow B}$ in order for Alice to transmit
a message $M$ to Bob. The receiver Bob combines his share $T_{B}$\ of the
entanglement and the $n$ output systems of the noisy channel in order to
decode the message. This scheme is an $\left(  n,R,\varepsilon\right)  $
protocol if the error probability is no larger than $\varepsilon>0$ and the
rate $R=\frac{1}{n}\log_{2}\left\vert M\right\vert $, where $\left\vert
M\right\vert $ denotes the size of the message.

For the purposes of proving a strong converse theorem, we can assume that
Alice selects the message $M$ according to a uniform distribution. (The rate
at which they can communicate when Alice uses a particular message
distribution can only be larger than that for a scheme that should work for
all message distributions.)\ Thus, the protocol begins with Alice preparing a
classically-correlated state of the following form:%
\begin{equation}
\overline{\Phi}_{MM^{^{\prime}}}=\frac{1}{|M|}\sum_{m}\left\vert
m\right\rangle \left\langle m\right\vert _{M}\otimes\left\vert m\right\rangle
\left\langle m\right\vert _{M^{^{\prime}}},
\end{equation}
and their goal will be for her and Bob to share a state close to this one at
the end of the protocol. Alice appends the registers $MM^{\prime}$\ to her
share $T_{A}$\ of the entanglement, so that the global state is%
\begin{equation}
\overline{\Phi}_{MM^{^{\prime}}}\otimes\Psi_{T_{A}T_{B}}.
\end{equation}
The most general encoding that she can perform is a CPTP\ map $\mathcal{E}%
_{M^{\prime}T_{A}\rightarrow A^{\prime n}}$, taking the $M^{\prime}T_{A}$
registers to a register $A^{\prime n}$ that can be transmitted through $n$
instances of the channel:%
\begin{align}
\rho_{MA^{\prime n}T_{B}} &  =\mathcal{E}_{M^{\prime}T_{A}\rightarrow
A^{\prime n}}\left(  \overline{\Phi}_{MM^{\prime}}\otimes\Psi_{T_{A}T_{B}%
}\right)  \\
&  =\frac{1}{\left\vert M\right\vert }\sum_{m}\left\vert m\right\rangle
\left\langle m\right\vert _{M}\otimes\mathcal{E}_{M^{\prime}T_{A}\rightarrow
A^{\prime n}}\left(  \left\vert m\right\rangle \left\langle m\right\vert
_{M^{\prime}}\otimes\Psi_{T_{A}T_{B}}\right)  \\
&  =\frac{1}{\left\vert M\right\vert }\sum_{m}\left\vert m\right\rangle
\left\langle m\right\vert _{M}\otimes\mathcal{E}_{T_{A}\rightarrow A^{\prime
n}}^{m}\left(  \Psi_{T_{A}T_{B}}\right)  .
\end{align}
In the above, we are exploiting the fact that a single CPTP\ map
$\mathcal{E}_{M^{\prime}T_{A}\rightarrow A^{\prime n}}$ acting on the
registers $M^{\prime}T_{A}$ can be written as $\left\vert M\right\vert
$~CPTP\ maps $\{\mathcal{E}_{T_{A}\rightarrow A^{\prime n}}^{m}\}$ acting on
the register $T_{A}$ (since $M$ is a classical register) \cite{W11}. Following
\cite{MW12}, we define the average code density operator as follows:%
\begin{equation}
\rho_{A^{\prime n}}\equiv\frac{1}{\left\vert M\right\vert }\sum_{m}%
\mathcal{E}_{T_{A}\rightarrow A^{\prime n}}^{m}\left(  \Psi_{T_{A}}\right)  .
\end{equation}
Alice then transmits the systems $A^{\prime n}$ over $n$ uses of the noisy
channel $\mathcal{N}_{A^{\prime}\rightarrow B}$, with the overall state
becoming%
\begin{align}
\omega_{MB^{n}T_{B}} &  \equiv\mathcal{N}_{A^{\prime n}\rightarrow B^{n}%
}({\rho}_{MA^{\prime n}T_{B}})\\
&  =\frac{1}{|M|}\sum_{m}\left\vert m\right\rangle \left\langle m\right\vert
_{M}\otimes\mathcal{N}_{A^{\prime n}\rightarrow B^{n}}(\mathcal{E}%
_{T_{A}{\rightarrow}A^{\prime n}}^{m}(\Psi_{T_{A}T_{B}})),
\end{align}
where $\mathcal{N}_{A^{\prime n}\rightarrow B^{n}}\equiv(\mathcal{N}%
_{A^{\prime}\rightarrow B})^{\otimes n}$. Bob performs a POVM $\{\Lambda
_{B^{n}T_{B}}^{m}\}$ on the registers $B^{n}T_{B}$ in order to estimate the
message $m$ sent by Alice. The overall state after this step is%
\begin{equation}
{\omega}_{M\hat{M}}\equiv\frac{1}{|M|}\sum_{m,m^{\prime}}\left\vert
m\right\rangle \left\langle m\right\vert _{M}\otimes\text{Tr}{\left\{
\Lambda_{B^{n}T_{B}}^{m^{\prime}}\mathcal{N}_{A^{\prime n}\rightarrow B^{n}%
}(\mathcal{E}_{T_{A}{\rightarrow}A^{\prime n}}^{m}(\Psi_{T_{A}T_{B}%
}))\right\}  }\left\vert m^{\prime}\right\rangle \left\langle m^{\prime
}\right\vert _{\hat{M}}.
\end{equation}
Thus, the average success probability for Bob detecting the message correctly
is given by%
\begin{align}
p_{\text{succ}} &  =\frac{1}{\left\vert M\right\vert }\sum_{m}\text{Pr}%
\left\{  {\hat{M}=m|M=m}\right\}  \\
&  =\frac{1}{\left\vert M\right\vert }\sum_{m}\text{Tr}{\left\{
\Lambda_{B^{n}T_{B}}^{m}\mathcal{N}_{A^{\prime n}\rightarrow B^{n}%
}(\mathcal{E}_{T_{A}{\rightarrow}A^{\prime n}}^{m}(\Psi_{T_{A}T_{B}%
}))\right\}  }%
\end{align}
For any $(n,R,\epsilon)$ protocol, the above success probability is bounded
from below by $1-\varepsilon$.

At this point, we recall the \textquotedblleft entanglement-assisted
meta-converse\textquotedblright\ \cite[Propositions~20 and 21]{MW12}:

\begin{proposition}
[{\cite[Propositions~20 and 21]{MW12}}]\label{prop:17MW}For any
entanglement-assisted code of the above form, with average code density
operator $\rho_{A^{\prime n}}$, there exists a two-outcome POVM\ $\left\{
T_{AB^{n}},I-T_{AB^{n}}\right\}  $ such that%
\begin{align}
p_{\operatorname{succ}} &  =\operatorname{Tr}\left\{  T_{AB^{n}}%
\mathcal{N}_{A^{\prime n}\rightarrow B^{n}}(\psi_{AA^{\prime n}}^{\rho
})\right\}  ,\\
\frac{1}{\left\vert M\right\vert } &  =\operatorname{Tr}\left\{  T_{AB^{n}%
}\left(  \psi_{A}^{\rho}\otimes\sigma_{B^{n}}\right)  \right\}  ,
\end{align}
where $\psi_{AA^{\prime n}}^{\rho}$ is a purification of the average code
density operator $\rho_{A^{\prime n}}$ and $\sigma_{B^{n}}$ is any density operator.
\end{proposition}

This proposition allows us to relate the rate and success probability of an
entanglement-assisted code to any generalized divergence, extending the
framework of \cite{PV10,SW12}\ to the entanglement-assisted case. Let
$I_{\mathbf{D}}\left(  \mathcal{M}\right)  $ denote the generalized mutual
information of a quantum channel $\mathcal{M}$, constructed from any
generalized divergence according to the recipe in (\ref{eq:recipe-first}%
)-(\ref{eq:recipe-last}). Let $\delta\left(  p\|q\right)  $ be equal to the
generalized divergence $\mathbf{D}\left(  \rho_{p}\|\rho_{q}\right)  $
evaluated for the commuting qubit states%
\begin{align}
\rho_{p}  &  =p\left\vert 0\right\rangle \left\langle 0\right\vert +\left(
1-p\right)  \left\vert 1\right\rangle \left\langle 1\right\vert ,\\
\rho_{q}  &  =q\left\vert 0\right\rangle \left\langle 0\right\vert +\left(
1-q\right)  \left\vert 1\right\rangle \left\langle 1\right\vert .
\end{align}
(Note that monotonicity of the generalized divergence implies its unitary
invariance, which in turn implies that it is independent of the basis when
evaluated for classical, commuting states.)

\begin{proposition}
\label{prop:gen-div-converse}The following bound holds for any $\left(
n,R,\varepsilon\right)  $ entanglement-assisted code:%
\begin{equation}
I_{\mathbf{D}}\left(  \mathcal{N}^{\otimes n}\right)  \geq\delta\left(
\varepsilon\|1-2^{-nR}\right)  .
\end{equation}

\end{proposition}

\begin{proof}
Our starting point for a proof is to exploit the test from
Proposition~\ref{prop:17MW} and define the classical states $\rho
_{p_{\text{succ}}}$ and $\rho_{1/\left\vert M\right\vert }$, which arise from
applying the measurement $\{T_{AB^{n}},I-T_{AB^{n}}\}$ to the following
states:%
\begin{align}
&  \mathcal{N}_{A^{\prime n}\rightarrow B^{n}}(\psi_{AA^{\prime n}}^{\rho}),\\
&  \psi_{A}^{\rho}\otimes\sigma_{B^{n}},
\end{align}
respectively. Without loss of generality, we can assume that $\varepsilon
\leq1-2^{-nR}$ (otherwise, there would be no need to prove the strong converse
since the error probability would obey the bound $\varepsilon>1-2^{-nR}$).
Then we have the following inequalities:%
\begin{align}
\delta\left(  \varepsilon\|1-2^{-nR}\right)   &  \leq\delta\left(
1-p_{\text{succ}}\middle\|1-2^{-nR}\right) \\
&  \leq\mathbf{D}\left(  \mathcal{N}_{A^{\prime n}\rightarrow B^{n}}%
(\psi_{AA^{\prime n}}^{\rho})\middle\|\psi_{A}^{\rho}\otimes\sigma_{B^{n}%
}\right)  .
\end{align}
The first inequality follows from the monotonicity $\delta\left(  p^{\prime
}\|q\right)  \leq\delta\left(  p\|q\right)  $ whenever $p\leq p^{\prime}\leq
q$ \cite{PV10} (recall that we have $1-p_{\text{succ}}\leq\varepsilon
\leq1-2^{-nR}$). The second inequality follows from monotonicity of the
generalized divergence under the test $\left\{  T_{AB^{n}},I-T_{AB^{n}%
}\right\}  $. Since the inequality holds for all states $\sigma_{B^{n}}$, we
can find the tightest upper bound on $\delta\left(  \varepsilon\|1-2^{-nR}%
\right)  $ for a code with average code density operator $\rho_{A^{\prime n}}$
by taking a minimum%
\begin{equation}
\delta\left(  \varepsilon\|1-2^{-nR}\right)  \leq\min_{\sigma_{B^{n}}%
}\mathbf{D}\left(  \mathcal{N}_{A^{\prime n}\rightarrow B^{n}}(\psi
_{AA^{\prime n}}^{\rho})\middle\|\psi_{A}^{\rho}\otimes\sigma_{B^{n}}\right)
.
\end{equation}
Finally, we can remove the dependence of the upper bound on any particular
code by maximizing over all code density operators $\rho_{A^{\prime n}}$:%
\begin{equation}
\delta\left(  \varepsilon\|1-2^{-nR}\right)  \leq\max_{\rho_{A^{\prime n}}%
}\min_{\sigma_{B^{n}}}\mathbf{D}\left(  \mathcal{N}_{A^{\prime n}\rightarrow
B^{n}}(\psi_{AA^{\prime n}}^{\rho})\middle\|\psi_{A}^{\rho}\otimes
\sigma_{B^{n}}\right)  .
\end{equation}
This is then equivalent to the inequality in the statement of the proposition.
\end{proof}

For our purposes here, we can evaluate the bound from
Proposition~\ref{prop:gen-div-converse} by setting the divergence to be the
sandwiched R\'{e}nyi relative entropy (however, note that there is no need for
the assumption $\varepsilon\leq1-2^{-nR}$ when employing the sandwiched
R\'{e}nyi relative entropy). Following steps identical to those in
\cite[Section~6]{WWY13}, we arrive at the following bound on the success
probability of any entanglement-assisted code:%
\begin{equation}
p_{\text{succ}}\leq2^{-n\left(  \frac{\alpha-1}{\alpha}\right)
\left(  R-\frac{1}{n}\widetilde{I}_{\alpha}\left(  \mathcal{N}^{\otimes
n}\right)  \right)  }.\label{eq:error-bound-a-holevo}%
\end{equation}
We stress that this bound holds for all $\alpha>1$ and $n\geq1$. In order to arrive at the
strong converse, we should now prove that the sandwiched R\'{e}nyi mutual
information is additive as a function of quantum channels for $\alpha
\in(1,\infty)$.

\section{Additivity of the sandwiched R\'{e}nyi mutual information of a
quantum channel}

In this section, we show that the sandwiched R\'{e}nyi mutual information
$\widetilde{I}_{\alpha}\left(  \mathcal{N}^{\otimes n}\right)  $ is additive
as a function of quantum channels for $\alpha\in(1,\infty)$. Before doing so,
we require a few supplementary lemmas.

\begin{lemma}
The following equality holds for $\alpha\in\left(  1,\infty\right)  $%
\begin{equation}
\widetilde{D}_{\alpha}\left(  \rho_{A}^{1/2}\Gamma_{AB}^{\mathcal{N}}\rho
_{A}^{1/2}\middle\|\rho_{A}\otimes\sigma_{B}\right)  =\frac{\alpha}{\alpha
-1}\log\left\Vert \left(  \left[  \Gamma_{AB}^{\mathcal{N}}\right]
^{1/2}\left(  \rho_{A}^{\frac{1}{\alpha}}\otimes\sigma_{B}^{\frac{1-\alpha
}{\alpha}}\right)  \left[  \Gamma_{AB}^{\mathcal{N}}\right]  ^{1/2}\right)
\right\Vert _{\alpha}%
\end{equation}

\end{lemma}

\begin{proof}
This follows from%
\begin{align}
&  \!\!\!\!\!\widetilde{D}_{\alpha}\left(  \rho_{A}^{1/2}\Gamma_{AB}%
^{\mathcal{N}}\rho_{A}^{1/2}\middle\|\rho_{A}\otimes\sigma_{B}\right)
\nonumber\\
&  =\frac{1}{\alpha-1}\log\text{Tr}\left\{  \left(  \left(  \rho_{A}%
\otimes\sigma_{B}\right)  ^{\frac{1-\alpha}{2\alpha}}\rho_{A}^{1/2}\Gamma
_{AB}^{\mathcal{N}}\rho_{A}^{1/2}\left(  \rho_{A}\otimes\sigma_{B}\right)
^{\frac{1-\alpha}{2\alpha}}\right)  ^{\alpha}\right\} \\
&  =\frac{1}{\alpha-1}\log\text{Tr}\left\{  \left(  \left(  \rho_{A}%
^{\frac{1-\alpha}{2\alpha}}\otimes\sigma_{B}^{\frac{1-\alpha}{2\alpha}%
}\right)  \rho_{A}^{1/2}\Gamma_{AB}^{\mathcal{N}}\rho_{A}^{1/2}\left(
\rho_{A}^{\frac{1-\alpha}{2\alpha}}\otimes\sigma_{B}^{\frac{1-\alpha}{2\alpha
}}\right)  \right)  ^{\alpha}\right\} \\
&  =\frac{1}{\alpha-1}\log\text{Tr}\left\{  \left(  \left(  \rho_{A}%
^{\frac{1-\alpha}{2\alpha}}\rho_{A}^{1/2}\otimes\sigma_{B}^{\frac{1-\alpha
}{2\alpha}}\right)  \Gamma_{AB}^{\mathcal{N}}\left(  \rho_{A}^{1/2}\rho
_{A}^{\frac{1-\alpha}{2\alpha}}\otimes\sigma_{B}^{\frac{1-\alpha}{2\alpha}%
}\right)  \right)  ^{\alpha}\right\} \\
&  =\frac{1}{\alpha-1}\log\text{Tr}\left\{  \left(  \left(  \rho_{A}^{\frac
{1}{2\alpha}}\otimes\sigma_{B}^{\frac{1-\alpha}{2\alpha}}\right)  \Gamma
_{AB}^{\mathcal{N}}\left(  \rho_{A}^{\frac{1}{2\alpha}}\otimes\sigma
_{B}^{\frac{1-\alpha}{2\alpha}}\right)  \right)  ^{\alpha}\right\} \\
&  =\frac{1}{\alpha-1}\log\text{Tr}\left\{  \left(  \left[  \Gamma
_{AB}^{\mathcal{N}}\right]  ^{1/2}\left(  \rho_{A}^{\frac{1}{\alpha}}%
\otimes\sigma_{B}^{\frac{1-\alpha}{\alpha}}\right)  \left[  \Gamma
_{AB}^{\mathcal{N}}\right]  ^{1/2}\right)  ^{\alpha}\right\} \\
&  =\frac{\alpha}{\alpha-1}\log\left\Vert \left(  \left[  \Gamma
_{AB}^{\mathcal{N}}\right]  ^{1/2}\left(  \rho_{A}^{\frac{1}{\alpha}}%
\otimes\sigma_{B}^{\frac{1-\alpha}{\alpha}}\right)  \left[  \Gamma
_{AB}^{\mathcal{N}}\right]  ^{1/2}\right)  \right\Vert _{\alpha}%
\end{align}

\end{proof}

\begin{lemma}
The following function is concave in $\rho_{A}\in\mathcal{S}\left(
\mathcal{H}_{A}\right)  $ and convex in $\sigma_B\in\mathcal{S}\left(
\mathcal{H}_{B}\right)  $ for $\alpha\in\left(  1,\infty\right)  $:%
\begin{equation}
\left\Vert \left(  \left[  \Gamma_{AB}^{\mathcal{N}}\right]  ^{1/2}\left(
\rho_{A}^{\frac{1}{\alpha}}\otimes\sigma_{B}^{\frac{1-\alpha}{\alpha}}\right)
\left[  \Gamma_{AB}^{\mathcal{N}}\right]  ^{1/2}\right)  \right\Vert _{\alpha
}.\label{eq:convex-concave-func}%
\end{equation}

\end{lemma}

\begin{proof}
Convexity in $\sigma_{B}$ follows immediately from operator convexity of
$x^{\left(  1-\alpha\right)  /\alpha}$ for $\alpha\in(1,\infty)$ and convexity
of the $a$-norm. The other statement follows in a few steps. We first
reexpress the $\alpha$-norm as the following optimization:%
\begin{equation}
\left\Vert \left(  \left[  \Gamma_{AB}^{\mathcal{N}}\right]  ^{1/2}\left(
\rho_{A}^{\frac{1}{\alpha}}\otimes\sigma_{B}^{\frac{1-\alpha}{\alpha}}\right)
\left[  \Gamma_{AB}^{\mathcal{N}}\right]  ^{1/2}\right)  \right\Vert _{\alpha
}=\max_{\left\Vert X\right\Vert _{\frac{\alpha}{\alpha-1}}\leq1}%
\text{Tr}\left\{  X\left[  \Gamma_{AB}^{\mathcal{N}}\right]  ^{1/2}\left(
\rho_{A}^{\frac{1}{\alpha}}\otimes\sigma_{B}^{\frac{1-\alpha}{\alpha}}\right)
\left[  \Gamma_{AB}^{\mathcal{N}}\right]  ^{1/2}\right\}  .
\end{equation}
Due to the operator $\left[  \Gamma_{AB}^{\mathcal{N}}\right]  ^{1/2}\left(
\rho_{A}^{\frac{1}{\alpha}}\otimes\sigma_{B}^{\frac{1-\alpha}{\alpha}}\right)
\left[  \Gamma_{AB}^{\mathcal{N}}\right]  ^{1/2}$ being positive, it suffices
to restrict the optimization to be over positive $X$ operators such that
$\left\Vert X\right\Vert _{\frac{\alpha}{\alpha-1}}\leq1$. However, this
restriction is equivalent to%
\begin{equation}
\text{Tr}\left\{  X^{\frac{\alpha}{\alpha-1}}\right\}  \leq1,\ \text{Tr}%
\left\{  X\right\}  \geq0.
\end{equation}
Thus, with the substitution $Y=X^{\alpha/\left(  \alpha-1\right)  }$, we can
rewrite the above as%
\begin{equation}
\max_{Y\geq0,\text{Tr}\left\{  Y\right\}  \leq1}\text{Tr}\left\{
Y^{\frac{\alpha-1}{\alpha}}\left[  \Gamma_{AB}^{\mathcal{N}}\right]
^{1/2}\left(  \rho_{A}^{\frac{1}{\alpha}}\otimes\sigma_{B}^{\frac{1-\alpha
}{\alpha}}\right)  \left[  \Gamma_{AB}^{\mathcal{N}}\right]  ^{1/2}\right\}  .
\end{equation}
To prove concavity in $\rho$, we require Lieb's concavity theorem \cite{L73},
a special case of which is the statement that the following function%
\begin{equation}
\text{Tr}\left\{  XR^{1-t}X^{\dag}S^{t}\right\}  ,
\end{equation}
for $R,S\geq0$ and $t\in\left[  0,1\right]  $, is jointly concave in $R$ and
$S$. Indeed, for $i\in\left\{  0,1\right\}  $, consider any $Y_{i}$ such that
$Y_{i}\geq0$, Tr$\left\{  Y_{i}\right\}  \leq1$, density operators $\rho_{i}$,
and $\lambda\in\left[  0,1\right]  $. We begin with%
\begin{align}
&  \lambda\text{Tr}\left\{  Y_{0}^{\frac{\alpha-1}{\alpha}}\left[  \Gamma
_{AB}^{\mathcal{N}}\right]  ^{1/2}\left(  \rho_{0}^{\frac{1}{\alpha}}%
\otimes\sigma_{B}^{\frac{1-\alpha}{\alpha}}\right)  \left[  \Gamma
_{AB}^{\mathcal{N}}\right]  ^{1/2}\right\}  \nonumber\\
&  \,\,\,\,\,\,\,\,\,\,\,+\left(  1-\lambda\right)  \text{Tr}\left\{
Y_{1}^{\frac{\alpha-1}{\alpha}}\left[  \Gamma_{AB}^{\mathcal{N}}\right]
^{1/2}\left(  \rho_{1}^{\frac{1}{\alpha}}\otimes\sigma_{B}^{\frac{1-\alpha
}{\alpha}}\right)  \left[  \Gamma_{AB}^{\mathcal{N}}\right]  ^{1/2}\right\}
\nonumber\\
&  \leq\text{Tr}\left\{  \left(  \lambda Y_{0}+\left(  1-\lambda\right)
Y_{1}\right)  ^{\frac{\alpha-1}{\alpha}}\left[  \Gamma_{AB}^{\mathcal{N}%
}\right]  ^{1/2}\left(  \left(  \lambda\rho_{0}+\left(  1-\lambda\right)
\rho_{1}\right)  ^{\frac{1}{\alpha}}\otimes\sigma_{B}^{\frac{1-\alpha}{\alpha
}}\right)  \left[  \Gamma_{AB}^{\mathcal{N}}\right]  ^{1/2}\right\}  \\
&  \leq\max_{Y\geq0,\text{Tr}\left\{  Y\right\}  \leq1}\text{Tr}\left\{
Y^{\frac{\alpha-1}{\alpha}}\left[  \Gamma_{AB}^{\mathcal{N}}\right]
^{1/2}\left(  \left(  \lambda\rho_{0}+\left(  1-\lambda\right)  \rho
_{1}\right)  ^{\frac{1}{\alpha}}\otimes\sigma_{B}^{\frac{1-\alpha}{\alpha}%
}\right)  \left[  \Gamma_{AB}^{\mathcal{N}}\right]  ^{1/2}\right\}  \\
&  =\left\Vert \left[  \Gamma_{AB}^{\mathcal{N}}\right]  ^{1/2}\left(  \left(
\lambda\rho_{0}+\left(  1-\lambda\right)  \rho_{1}\right)  ^{\frac{1}{\alpha}%
}\otimes\sigma_{B}^{\frac{1-\alpha}{\alpha}}\right)  \left[  \Gamma
_{AB}^{\mathcal{N}}\right]  ^{1/2}\right\Vert _{\alpha}%
\end{align}
Since the calculation is independent of which $Y_{0}$ and $Y_{1}$ we started
with, concavity of (\ref{eq:convex-concave-func}) in $\rho$ follows.
\end{proof}

\begin{lemma}
The following equality holds for $\alpha\in\left(  1,\infty\right)  $%
\begin{equation}
\widetilde{I}_{\alpha}\left(  \mathcal{N}\right)  =\min_{\sigma_{B}}\max
_{\rho_{A}}\widetilde{D}_{\alpha}\left(  \rho_{A}^{1/2}\Gamma_{AB}%
^{\mathcal{N}}\rho_{A}^{1/2}\middle\|\rho_{A}\otimes\sigma_{B}\right)  .
\end{equation}

\end{lemma}

\begin{proof}
Consider that%
\begin{align}
\widetilde{I}_{\alpha}\left(  \mathcal{N}\right)   &  =\max_{\rho_{A}}%
\min_{\sigma_{B}}\widetilde{D}_{\alpha}\left(  \rho_{A}^{1/2}\Gamma
_{AB}^{\mathcal{N}}\rho_{A}^{1/2}\middle\|\rho_{A}\otimes\sigma_{B}\right) \\
&  =\max_{\rho_{A}}\min_{\sigma_{B}}\frac{\alpha}{\alpha-1}\log\left\Vert
\left(  \left[  \Gamma_{AB}^{\mathcal{N}}\right]  ^{1/2}\left(  \rho
_{A}^{\frac{1}{\alpha}}\otimes\sigma_{B}^{\frac{1-\alpha}{\alpha}}\right)
\left[  \Gamma_{AB}^{\mathcal{N}}\right]  ^{1/2}\right)  \right\Vert _{\alpha
}\\
&  =\frac{\alpha}{\alpha-1}\log\max_{\rho_{A}}\min_{\sigma_{B}}\left\Vert
\left(  \left[  \Gamma_{AB}^{\mathcal{N}}\right]  ^{1/2}\left(  \rho
_{A}^{\frac{1}{\alpha}}\otimes\sigma_{B}^{\frac{1-\alpha}{\alpha}}\right)
\left[  \Gamma_{AB}^{\mathcal{N}}\right]  ^{1/2}\right)  \right\Vert _{\alpha
}\\
&  =\frac{\alpha}{\alpha-1}\log\min_{\sigma_{B}}\max_{\rho_{A}}\left\Vert
\left(  \left[  \Gamma_{AB}^{\mathcal{N}}\right]  ^{1/2}\left(  \rho
_{A}^{\frac{1}{\alpha}}\otimes\sigma_{B}^{\frac{1-\alpha}{\alpha}}\right)
\left[  \Gamma_{AB}^{\mathcal{N}}\right]  ^{1/2}\right)  \right\Vert _{\alpha
}\\
&  =\min_{\sigma_{B}}\max_{\rho_{A}}\frac{\alpha}{\alpha-1}\log\left\Vert
\left(  \left[  \Gamma_{AB}^{\mathcal{N}}\right]  ^{1/2}\left(  \rho
_{A}^{\frac{1}{\alpha}}\otimes\sigma_{B}^{\frac{1-\alpha}{\alpha}}\right)
\left[  \Gamma_{AB}^{\mathcal{N}}\right]  ^{1/2}\right)  \right\Vert _{\alpha
}\\
&  =\min_{\sigma_{B}}\max_{\rho_{A}}\widetilde{D}_{\alpha}\left(  \rho
_{A}^{1/2}\Gamma_{AB}^{\mathcal{N}}\rho_{A}^{1/2}\middle\|\rho_{A}%
\otimes\sigma_{B}\right)  ,
\end{align}
where we applied the Sion minimax theorem \cite{S58}.
\end{proof}

We now prove additivity by exploiting the above lemmas and some results in
\cite{B13monotone,DJKR06}:

\begin{lemma}
\label{lem:subadd-EA}The sandwiched R\'{e}nyi mutual information is additive
as a function of channels for $\alpha\in\left(  1,\infty\right)  $, in the
sense that%
\begin{equation}
\widetilde{I}_{\alpha}\left(  \mathcal{N}_{1}\otimes\mathcal{N}_{2}\right)
=\widetilde{I}_{\alpha}\left(  \mathcal{N}_{1}\right)  +\widetilde{I}_{\alpha
}\left(  \mathcal{N}_{2}\right)  .
\end{equation}

\end{lemma}

\begin{proof}
The inequality below is straightforward%
\begin{equation}
\widetilde{I}_{\alpha}\left(  \mathcal{N}_{1}\otimes\mathcal{N}_{2}\right)
\geq\widetilde{I}_{\alpha}\left(  \mathcal{N}_{1}\right)  +\widetilde
{I}_{\alpha}\left(  \mathcal{N}_{2}\right)  ,
\end{equation}
following from \cite[Theorem~11]{B13monotone} and the fact that we can choose
tensor-product states as a special case of the optimization on the left hand side.

We now prove the other inequality:%
\begin{equation}
\widetilde{I}_{\alpha}\left(  \mathcal{N}_{1}\otimes\mathcal{N}_{2}\right)
\leq\widetilde{I}_{\alpha}\left(  \mathcal{N}_{1}\right)  +\widetilde
{I}_{\alpha}\left(  \mathcal{N}_{2}\right)  .
\end{equation}
From the above lemmas, we can reexpress $\widetilde{I}_{\alpha}\left(
\mathcal{N}\right)  $\ as%
\begin{equation}
\min_{\sigma_{B}}\max_{\rho_{A}}\frac{\alpha}{\alpha-1}\log\left\Vert
\sigma_{B}^{\frac{1-\alpha}{2\alpha}}\mathcal{N}_{A^{\prime}\rightarrow
B}\left(  \rho_{A}^{\frac{1}{2\alpha}}\Gamma_{AA^{\prime}}\rho_{A}^{\frac
{1}{2\alpha}}\right)  \sigma_{B}^{\frac{1-\alpha}{2\alpha}}\right\Vert
_{\alpha}.
\end{equation}
Defining the CP\ map $\Theta_{\sigma}\left(  X\right)  =\sigma^{1/2}%
X\sigma^{1/2}$, we can write the above as%
\begin{multline}
\min_{\sigma_{B}}\frac{\alpha}{\alpha-1}\log\max_{\rho_{A}}\left\Vert \left(
\Theta_{\sigma^{\frac{1-\alpha}{\alpha}}}\circ\mathcal{N}_{A^{\prime
}\rightarrow B}\right)  \left(  \rho_{A}^{\frac{1}{2\alpha}}\Gamma
_{AA^{\prime}}\rho_{A}^{\frac{1}{2\alpha}}\right)  \right\Vert _{\alpha}\\
=\min_{\sigma_{B}}\frac{\alpha}{\alpha-1}\log\left\Vert \Theta_{\sigma
^{\frac{1-\alpha}{\alpha}}}\circ\mathcal{N}_{A^{\prime}\rightarrow
B}\right\Vert _{\text{CB},1\rightarrow\alpha}.
\end{multline}
The equality follows from \cite[Theorem~10]{DJKR06}, in which these authors
showed that the $\left(  \text{CB},1\rightarrow\alpha\right)  $\ norm of a
CP\ map $\mathcal{M}_{A^{\prime}\rightarrow B}$ is equal to%
\begin{align}
\left\Vert \mathcal{M}_{A^{\prime}\rightarrow B}\right\Vert _{\text{CB}%
,1\rightarrow\alpha} &  \equiv\sup_{X>0}\frac{\left\Vert \left(  X\otimes
I\right)  \mathcal{M}_{A^{\prime}\rightarrow B}\left(  \Gamma_{AA^{\prime}%
}\right)  \left(  X\otimes I\right)  \right\Vert _{\alpha}}{\left\Vert
X^{2}\right\Vert _{\alpha}}\\
&  =\sup_{Y>0,\text{Tr}\left\{  Y\right\}  \leq1}\left\Vert \left(
Y^{\frac{1}{2\alpha}}\otimes I\right)  \mathcal{M}_{A^{\prime}\rightarrow
B}\left(  \Gamma_{AA^{\prime}}\right)  \left(  Y^{\frac{1}{2\alpha}}\otimes
I\right)  \right\Vert _{\alpha}%
\end{align}
With this result in hand, defining%
\begin{equation}
\Gamma_{A_{1}A_{2}B_{1}B_{2}}^{\mathcal{N}_{1}\otimes\mathcal{N}_{1}}%
\equiv\left(  \mathcal{N}_{1}\otimes\mathcal{N}_{2}\right)  \left(
\Gamma_{A_{1}A_{2}A_{1}^{\prime}A_{2}^{\prime}}\right)  ,
\end{equation}
we can now prove the other inequality:%
\begin{align}
\widetilde{I}_{\alpha}\left(  \mathcal{N}_{1}\otimes\mathcal{N}_{2}\right)
&  =\max_{\rho_{A_{1}A_{2}}}\min_{\sigma_{B_{1}B_{2}}}\widetilde{D}_{\alpha
}\left(  \rho_{A_{1}A_{2}}^{1/2}\Gamma_{A_{1}A_{2}B_{1}B_{2}}^{\mathcal{N}%
_{1}\otimes\mathcal{N}_{1}}\rho_{A_{1}A_{2}}^{1/2}\middle\Vert\rho_{A_{1}%
A_{2}}\otimes\sigma_{B_{1}B_{2}}\right)  \\
&  =\min_{\sigma_{B_{1}B_{2}}}\max_{\rho_{A_{1}A_{2}}}\widetilde{D}_{\alpha
}\left(  \rho_{A_{1}A_{2}}^{1/2}\Gamma_{A_{1}A_{2}B_{1}B_{2}}^{\mathcal{N}%
_{1}\otimes\mathcal{N}_{1}}\rho_{A_{1}A_{2}}^{1/2}\middle\Vert\rho_{A_{1}%
A_{2}}\otimes\sigma_{B_{1}B_{2}}\right)  \\
&  \leq\min_{\sigma_{B_{1}}\otimes\sigma_{B_{2}}}\max_{\rho_{A_{1}A_{2}}%
}\widetilde{D}_{\alpha}\left(  \rho_{A_{1}A_{2}}^{1/2}\Gamma_{A_{1}A_{2}%
B_{1}B_{2}}^{\mathcal{N}_{1}\otimes\mathcal{N}_{1}}\rho_{A_{1}A_{2}}%
^{1/2}\middle\Vert\rho_{A_{1}A_{2}}\otimes\sigma_{B_{1}}\otimes\sigma_{B_{2}%
}\right)  \\
&  =\frac{\alpha}{\alpha-1}\log\min_{\sigma_{B_{1}}\otimes\sigma_{B_{2}}%
}\left\Vert \left(  \Theta_{\sigma_{B_{1}}^{\frac{1-\alpha}{\alpha}}}%
\circ\mathcal{N}_{1}\right)  \otimes\left(  \Theta_{\sigma_{B_{2}}%
^{\frac{1-\alpha}{\alpha}}}\circ\mathcal{N}_{2}\right)  \right\Vert
_{\text{CB},1\rightarrow\alpha}%
\end{align}
Using \cite[Theorem 11]{DJKR06} (the main result there), we find that the
above is equal to%
\begin{align}
&  =\frac{\alpha}{\alpha-1}\log\min_{\sigma_{B_{1}}\otimes\sigma_{B_{2}}%
}\left\Vert \left(  \Theta_{\sigma_{B_{1}}^{\frac{1-\alpha}{\alpha}}}%
\circ\mathcal{N}_{1}\right)  \right\Vert _{\text{CB},1\rightarrow\alpha
}\left\Vert \left(  \Theta_{\sigma_{B_{2}}^{\frac{1-\alpha}{\alpha}}}%
\circ\mathcal{N}_{2}\right)  \right\Vert _{\text{CB},1\rightarrow\alpha}\\
&  =\frac{\alpha}{\alpha-1}\log\min_{\sigma_{B_{1}}}\left\Vert \left(
\Theta_{\sigma_{B_{1}}^{\frac{1-\alpha}{\alpha}}}\circ\mathcal{N}_{1}\right)
\right\Vert _{\text{CB},1\rightarrow\alpha}\min_{\sigma_{B_{2}}}\left\Vert
\left(  \Theta_{\sigma_{B_{2}}^{\frac{1-\alpha}{\alpha}}}\circ\mathcal{N}%
_{2}\right)  \right\Vert _{\text{CB},1\rightarrow\alpha}\\
&  =\frac{\alpha}{\alpha-1}\log\min_{\sigma_{B_{1}}}\left\Vert \left(
\Theta_{\sigma_{B_{1}}^{\frac{1-\alpha}{\alpha}}}\circ\mathcal{N}_{1}\right)
\right\Vert _{\text{CB},1\rightarrow\alpha}+\frac{\alpha}{\alpha-1}\log
\min_{\sigma_{B_{2}}}\left\Vert \left(  \Theta_{\sigma_{B_{2}}^{\frac
{1-\alpha}{\alpha}}}\circ\mathcal{N}_{2}\right)  \right\Vert _{\text{CB}%
,1\rightarrow\alpha}\\
&  =\widetilde{I}_{\alpha}\left(  \mathcal{N}_{1}\right)  +\widetilde
{I}_{\alpha}\left(  \mathcal{N}_{2}\right)  .
\end{align}

\end{proof}

\section{Final steps for the strong converse}

\label{sec:final-steps-EB}In this section, we outline the remaining steps to
prove the strong converse theorem. Returning to~(\ref{eq:error-bound-a-holevo}%
), the additivity relation from Lemma~\ref{lem:subadd-EA} (along with an
inductive argument) allows us to conclude the following upper bound on the
success probability when using an entanglement-assisted code to communicate
over a quantum channel $\mathcal{N}$:%
\begin{equation}
p_{\text{succ}}\leq2^{-n\sup_{\alpha>1}\left(  \frac{\alpha-1}{\alpha}\right)
\left(  R-\widetilde{I}_{\alpha}\left(  \mathcal{N}\right)  \right)
}.\label{eq:improved-bound-succ-prob}%
\end{equation}
The quantity $\sup_{\alpha>1}\left(  \frac{\alpha-1}{\alpha}\right)  \left(
R-\widetilde{I}_{\alpha}\left(  \mathcal{N}\right)  \right)  $ is thus our
bound on the strong converse exponent, which holds for all $n\geq1$. 

Recall the definition of the (traditional) quantum R\'{e}nyi relative entropy
for $\alpha\in\left(  1,\infty\right)  $:%
\begin{equation}
D_{\alpha}\left(  \rho\Vert\sigma\right)  \equiv\frac{1}{\alpha-1}%
\log\text{Tr}\left\{  \rho^{\alpha}\sigma^{1-\alpha}\right\}
,\label{eq:traditional-renyi-entropy}%
\end{equation}
whenever the support of $\rho$ is in the support of $\sigma$ and it is equal
to $+\infty$ otherwise. We define the R\'{e}nyi quantum mutual information of a
bipartite state $\rho_{AB}$ as follows:%
\begin{equation}
{I}_{\alpha}\left(  A;B\right)  _{\rho}=\min_{\sigma_{B}}D_{\alpha}\left(
\rho_{AB}\Vert\rho_{A}\otimes\sigma_{B}\right)  .
\end{equation}

By applying the following Lieb-Thirring trace inequality \cite{LT76,Carlen09},
which holds for $B\geq0$, any operator $C$, and for $\alpha\geq1$:%
\begin{equation}
\text{Tr}\{(CBC^{\dag})^{\alpha}\}\leq\text{Tr}\{(C^{\dag}C)^{\alpha}%
B^{\alpha}\},\label{eq:lieb-thirring}%
\end{equation}
we find that the following inequality holds for $\alpha>1$ \cite{WWY13}:%
\begin{equation}
\widetilde{D}_{\alpha}\left(  \rho\Vert\sigma\right)  \leq D_{\alpha}\left(
\rho\Vert\sigma\right)  .\label{eq:ineq-d-tilde-alpha-d-alpha}%
\end{equation}
This in turn implies the following upper bound on the success probability of
any entanglement-assisted code for all $\alpha\in(1,\infty)$ by combining
(\ref{eq:ineq-d-tilde-alpha-d-alpha}) with (\ref{eq:improved-bound-succ-prob}%
):
\begin{equation}
p_{\text{succ}}\leq2^{-n\left(  \frac{\alpha-1}{\alpha}\right)  \left(
R-{I}_{\alpha}\left(  \mathcal{N}\right)  \right)  },\label{eq:new-succ-prob}%
\end{equation}
where we define ${I}_{\alpha}\left(  \mathcal{N}\right)  $ according to the
recipe in (\ref{eq:recipe-first})-(\ref{eq:recipe-last}) with the divergence
set to (\ref{eq:traditional-renyi-entropy}).

We next prove the following ``quantum Sibson identity,'' which will be helpful
in obtaining an explicit form for the R\'enyi quantum mutual information (see
\cite{SW12} for a variant which is relevant for R\'enyi coherent information):

\begin{lemma}
[Quantum Sibson identity]\label{lem:q-sibson}The following quantum Sibson
identity holds for $\alpha\in(1,\infty)$%
\begin{equation}
D_{\alpha}\left(  \rho_{AB}\Vert\rho_{A}\otimes\sigma_{B}\right)  =D_{\alpha
}\left(  \sigma_{B}^{\ast}\Vert\sigma_{B}\right)  +D_{\alpha}\left(  \rho
_{AB}\Vert\rho_{A}\otimes\sigma_{B}^{\ast}\right)  ,\label{eq:sibson}%
\end{equation}
where $\sigma_{B}^{\ast}$ is defined as%
\begin{equation}
\sigma_{B}^{\ast}\equiv\frac{\left(  \operatorname{Tr}_{A}\left\{  \rho
_{A}^{1-\alpha}\rho_{AB}^{\alpha}\right\}  \right)  ^{1/\alpha}}%
{\operatorname{Tr}\left\{  \left(  \operatorname{Tr}_{A}\left\{  \rho
_{A}^{1-\alpha}\rho_{AB}^{\alpha}\right\}  \right)  ^{1/\alpha}\right\}  }\ .
\end{equation}

\end{lemma}

\begin{proof}
It is clear that $\sigma_{B}^{\ast}$ is a positive operator because $\rho
_{A}^{\frac{1-\alpha}{2}}\rho_{AB}^{\alpha}\rho_{A}^{\frac{1-\alpha}{2}}$ is
positive, and the partial trace maintains positivity while being equal to
Tr$_{A}\left\{  \rho_{A}^{1-\alpha}\rho_{AB}^{\alpha}\right\}  $ from
cyclicity. The above relation then implies that%
\begin{equation}
\left(  \sigma_{B}^{\ast}\text{Tr}\left\{  \left(  \text{Tr}_{A}\left\{
\rho_{A}^{1-\alpha}\rho_{AB}^{\alpha}\right\}  \right)  ^{1/\alpha}\right\}
\right)  ^{\alpha}=\text{Tr}_{A}\left\{  \rho_{A}^{1-\alpha}\rho_{AB}^{\alpha
}\right\}  .
\end{equation}
We can then expand $D_{\alpha}\left(  \rho_{AB}\Vert\rho_{A}\otimes\sigma
_{B}\right)  $ as follows:%
\begin{align}
D_{\alpha}\left(  \rho_{AB}\Vert\rho_{A}\otimes\sigma_{B}\right)   &
=\frac{1}{\alpha-1}\log\text{Tr}\left\{  \rho_{AB}^{\alpha}\left(  \rho
_{A}^{1-\alpha}\otimes\sigma_{B}^{1-\alpha}\right)  \right\}  \\
&  =\frac{1}{\alpha-1}\log\text{Tr}\left\{  \text{Tr}_{A}\left\{  \rho
_{A}^{1-\alpha}\rho_{AB}^{\alpha}\right\}  \ \sigma_{B}^{1-\alpha}\right\}  \\
&  =\frac{1}{\alpha-1}\log\text{Tr}\left\{  \left(  \sigma_{B}^{\ast
}\ \text{Tr}\left\{  \left(  \text{Tr}_{A}\left\{  \rho_{A}^{1-\alpha}%
\rho_{AB}^{\alpha}\right\}  \right)  ^{1/\alpha}\right\}  \right)  ^{\alpha
}\ \sigma_{B}^{1-\alpha}\right\}  \\
&  =\frac{1}{\alpha-1}\left[  \log\text{Tr}\left\{  \left(  \sigma_{B}^{\ast
}\right)  ^{\alpha}\ \sigma_{B}^{1-\alpha}\right\}  +\alpha\log\text{Tr}%
\left\{  \left(  \text{Tr}_{A}\left\{  \rho_{A}^{1-\alpha}\rho_{AB}^{\alpha
}\right\}  \right)  ^{1/\alpha}\right\}  \right]  \\
&  =D_{\alpha}\left(  \sigma_{B}^{\ast}\Vert\sigma_{B}\right)  +\frac{\alpha
}{\alpha-1}\log\text{Tr}\left\{  \left(  \text{Tr}_{A}\left\{  \rho
_{A}^{1-\alpha}\rho_{AB}^{\alpha}\right\}  \right)  ^{1/\alpha}\right\}
.\label{eq:sibson-1}%
\end{align}
Now consider expanding $D_{\alpha}\left(  \rho_{AB}\Vert\rho_{A}\otimes
\sigma_{B}^{\ast}\right)  $:%
\begin{align}
&  D_{\alpha}\left(  \rho_{AB}\Vert\rho_{A}\otimes\sigma_{B}^{\ast}\right)
\nonumber\\
&  =\frac{1}{\alpha-1}\log\text{Tr}\left\{  \text{Tr}_{A}\left\{  \rho
_{A}^{1-\alpha}\rho_{AB}^{\alpha}\right\}  \ \left(  \sigma_{B}^{\ast}\right)
^{1-\alpha}\right\}  \\
&  =\frac{1}{\alpha-1}\log\text{Tr}\left\{  \text{Tr}_{A}\left\{  \rho
_{A}^{1-\alpha}\rho_{AB}^{\alpha}\right\}  \ \left(  \left(  \text{Tr}%
_{A}\left\{  \rho_{A}^{1-\alpha}\rho_{AB}^{\alpha}\right\}  \right)
^{1/\alpha}\right)  ^{1-\alpha}\right\}  +\log\left(  \text{Tr}\left\{
\left(  \text{Tr}_{A}\left\{  \rho_{A}^{1-\alpha}\rho_{AB}^{\alpha}\right\}
\right)  ^{1/\alpha}\right\}  \right)  \\
&  =\frac{1}{\alpha-1}\log\text{Tr}\left\{  \text{Tr}_{A}\left\{  \rho
_{A}^{1-\alpha}\rho_{AB}^{\alpha}\right\}  \ \left(  \text{Tr}_{A}\left\{
\rho_{A}^{1-\alpha}\rho_{AB}^{\alpha}\right\}  \right)  ^{\frac{1-\alpha
}{\alpha}}\right\}  +\log\left(  \text{Tr}\left\{  \left(  \text{Tr}%
_{A}\left\{  \rho_{A}^{1-\alpha}\rho_{AB}^{\alpha}\right\}  \right)
^{1/\alpha}\right\}  \right)  \\
&  =\frac{1}{\alpha-1}\log\text{Tr}\left\{  \left(  \text{Tr}_{A}\left\{
\rho_{A}^{1-\alpha}\rho_{AB}^{\alpha}\right\}  \right)  ^{\frac{1}{\alpha}%
}\right\}  +\log\left(  \text{Tr}\left\{  \left(  \text{Tr}_{A}\left\{
\rho_{A}^{1-\alpha}\rho_{AB}^{\alpha}\right\}  \right)  ^{1/\alpha}\right\}
\right)  \\
&  =\frac{\alpha}{\alpha-1}\log\text{Tr}\left\{  \left(  \text{Tr}_{A}\left\{
\rho_{A}^{1-\alpha}\rho_{AB}^{\alpha}\right\}  \right)  ^{\frac{1}{\alpha}%
}\right\}  .\label{eq:sibson-2}%
\end{align}
Combining (\ref{eq:sibson-1}) and (\ref{eq:sibson-2}) gives (\ref{eq:sibson}).
\end{proof}

By exploiting Lemma~\ref{lem:q-sibson}, we obtain an explicit form for it:

\begin{corollary}
\label{lem:Renyi-mutual-explicit-relation}The R\'{e}nyi quantum mutual
information has an explicit form for $\alpha\in(1,\infty)$ given by%
\begin{equation}
{I}_{\alpha}\left(  A;B\right)  _{\rho}=\frac{\alpha}{\alpha-1}\log
\operatorname{Tr}\left\{  \left(  \operatorname{Tr}_{A}\left\{  \rho
_{A}^{1-\alpha}\rho_{AB}^{\alpha}\right\}  \right)  ^{\frac{1}{\alpha}%
}\right\}
\end{equation}

\end{corollary}

\begin{proof}
From the identity in Lemma~\ref{lem:q-sibson}, we can conclude that%
\begin{align}
{I}_{\alpha}\left(  A;B\right)  _{\rho}  &  =\min_{\sigma_{B}}D_{\alpha
}\left(  \rho_{AB}\|\rho_{A}\otimes\sigma_{B}\right) \\
&  =\min_{\sigma_{B}}\left[  D_{\alpha}\left(  \sigma_{B}^{\ast}\|\sigma
_{B}\right)  +D_{\alpha}\left(  \rho_{AB}\|\rho_{A}\otimes\sigma_{B}^{\ast
}\right)  \right] \\
&  =D_{\alpha}\left(  \rho_{AB}\|\rho_{A}\otimes\sigma_{B}^{\ast}\right) \\
&  =\frac{\alpha}{\alpha-1}\log\text{Tr}\left\{  \left(  \text{Tr}_{A}\left\{
\rho_{A}^{1-\alpha}\rho_{AB}^{\alpha}\right\}  \right)  ^{\frac{1}{\alpha}%
}\right\}  .
\end{align}

\end{proof}

The following lemma will be helpful for us in relating R\'enyi quantum mutual
information to the von Neumann quantum mutual information:

\begin{lemma}
\label{lem:converge-to-I} The following identity holds for a bipartite state
$\rho_{AB}$:%
\begin{equation}
\lim_{\alpha\searrow1}\left[  \frac{\partial}{\partial\alpha}\log
\operatorname{Tr}\left\{  \left(  \operatorname{Tr}_{A}\left\{  \rho
_{A}^{1-\alpha}\rho_{AB}^{\alpha}\right\}  \right)  ^{\frac{1}{\alpha}%
}\right\}  \right]  =I\left(  A;B\right)  _{\rho}\ .
\end{equation}

\end{lemma}

\begin{proof}
A proof follows by exploiting some ideas from \cite{CL08} and \cite{ON99}. It
suffices to show that%
\begin{equation}
\lim_{\alpha\searrow1}\left[  \frac{\partial}{\partial\alpha}\log
\operatorname{Tr}\left\{  \left(  \operatorname{Tr}_{A}\left\{  \rho
_{A}^{1-\alpha}\rho_{AB}^{\alpha}\right\}  \right)  ^{\frac{1}{\alpha}%
}\right\}  \right]  =-\text{Tr}\left\{  \rho_{A}\log\rho_{A}\right\}
-\text{Tr}\left\{  \rho_{B}\log\rho_{B}\right\}  +\text{Tr}\left\{  \rho
_{AB}\log\rho_{AB}\right\}  .
\end{equation}
(In this proof, we will take $\log$ to denote the natural logarithm, but note
that the result follows simply by replacing the natural logarithm in both
definitions with the binary logarithm.)

Let us rewrite the expression inside the trace, using $\alpha=1+\beta$ where
$\beta>0$, as%
\begin{equation}
\text{Tr}\left\{  \left(  \text{Tr}_{A}\left\{  \rho_{A}^{-\beta}\rho
_{AB}^{1+\beta}\right\}  \right)  ^{\frac{1}{1+\beta}}\right\}  .
\end{equation}
Furthermore, we can introduce two parameters $\beta_{1}>0$ and $\beta_{2}>0$,
so that the above expression is a special case of%
\begin{equation}
f\left(  \beta_{1},\beta_{2}\right)  \equiv\text{Tr}\left\{  \left(
\text{Tr}_{A}\left\{  \rho_{A}^{-\beta_{1}}\rho_{AB}^{1+\beta_{1}}\right\}
\right)  ^{\frac{1}{1+\beta_{2}}}\right\}  .
\end{equation}
We then have that%
\begin{align}
\lim_{\alpha\searrow1}\left[  \frac{\partial}{\partial\alpha}\log
\text{Tr}\left\{  \left(  \text{Tr}_{A}\left\{  \rho_{A}^{1-\alpha}\rho
_{AB}^{\alpha}\right\}  \right)  ^{\frac{1}{\alpha}}\right\}  \right]   &
=\frac{\lim_{\beta\searrow0}\left[  \frac{\partial}{\partial\beta}f\left(
\beta,\beta\right)  \right]  }{f\left(  0,0\right)  }\\
&  =\lim_{\beta_{1}\searrow0}\left[  \frac{\partial}{\partial\beta_{1}%
}f\left(  \beta_{1},0\right)  \right]  +\lim_{\beta_{2}\searrow0}\left[
\frac{\partial}{\partial\beta_{2}}f\left(  0,\beta_{2}\right)  \right]  ,
\end{align}
where the second equality follows in part because $f\left(  0,0\right)  =1$.
Now consider the following Taylor expansions around~$\beta=0$:
\begin{align}
X^{-\beta} &  =I-\beta\log X+O\left(  \beta^{2}\right)  ,\\
X^{1+\beta} &  =X+\beta X\log X+O\left(  \beta^{2}\right)  ,\\
X^{\frac{1}{1+\beta}} &  =X-\beta X\log X+O\left(  \beta^{2}\right)  .
\end{align}
From these, we calculate $f\left(  \beta_{1},0\right)  $ as%
\begin{align}
f\left(  \beta_{1},0\right)   &  =\text{Tr}\left\{  \rho_{A}^{-\beta_{1}}%
\rho_{AB}^{1+\beta_{1}}\right\}  \\
&  =\text{Tr}\left\{  \rho_{AB}-\beta_{1}\rho_{AB}\log\rho_{A}+\beta_{1}%
\rho_{AB}\log\rho_{AB}\right\}  +O\left(  \beta_{1}^{2}\right)  \\
&  =\text{Tr}\left\{  \rho_{AB}\right\}  -\beta_{1}\text{Tr}\left\{  \rho
_{AB}\log\rho_{A}\right\}  +\beta_{1}\text{Tr}\left\{  \rho_{AB}\log\rho
_{AB}\right\}  +O\left(  \beta_{1}^{2}\right)  .
\end{align}
It then follows that%
\begin{equation}
\lim_{\beta_{1}\searrow0}\left[  \frac{\partial}{\partial\beta_{1}}f\left(
\beta_{1},0\right)  \right]  =-\text{Tr}\left\{  \rho_{A}\log\rho_{A}\right\}
+\text{Tr}\left\{  \rho_{AB}\log\rho_{AB}\right\}  .
\end{equation}
We then calculate $f\left(  0,\beta_{2}\right)  $ as%
\begin{align}
f\left(  0,\beta_{2}\right)   &  =\text{Tr}\left\{  \left(  \text{Tr}%
_{A}\left\{  \rho_{AB}\right\}  \right)  ^{\frac{1}{1+\beta_{2}}}\right\}  \\
&  =\text{Tr}\left\{  \left(  \rho_{B}\right)  ^{\frac{1}{1+\beta_{2}}%
}\right\}  \\
&  =\text{Tr}\left\{  \rho_{B}\right\}  -\beta_{2}\text{Tr}\left\{  \rho
_{B}\log\rho_{B}\right\}  +O\left(  \beta_{2}^{2}\right)  .
\end{align}
It then follows that%
\begin{equation}
\lim_{\beta_{2}\searrow0}\left[  \frac{\partial}{\partial\beta_{2}}f\left(
0,\beta_{2}\right)  \right]  =-\text{Tr}\left\{  \rho_{B}\log\rho_{B}\right\}
.
\end{equation}
Putting these together, we find that%
\begin{equation}
\lim_{\beta\searrow0}\left[  \frac{\partial}{\partial\beta}f\left(
\beta,\beta\right)  \right]  =-\text{Tr}\left\{  \rho_{A}\log\rho_{A}\right\}
-\text{Tr}\left\{  \rho_{B}\log\rho_{B}\right\}  +\text{Tr}\left\{  \rho
_{AB}\log\rho_{AB}\right\}  =I\left(  A;B\right)  _{\rho}.
\end{equation}

\end{proof}

Let ${I}(\mathcal{N})$ denote the entanglement-assisted capacity of a quantum
channel $\mathcal{N}$, which \cite{Bennett2002} proved is a function of
$\mathcal{N}$
and constructed according to the recipe in (\ref{eq:recipe-first}%
)-(\ref{eq:recipe-last}) with the generalized divergence taken as the quantum
relative entropy $D\left(  \rho\Vert\sigma\right)  $.

\begin{lemma}
\label{lem: rate-capacity-gt-zero} If $R>{I}(\mathcal{N})$ then
\begin{equation}
\exists\;\beta>1,\,\forall\alpha\in\left(  1,\beta\right)  ,\qquad\left(
\frac{\alpha-1}{\alpha}\right)  \left(  R-{I}_{\alpha}\left(  \mathcal{N}%
\right)  \right)  >0. \label{eq:fun-ge-than-zero}%
\end{equation}

\end{lemma}

\begin{proof}
The argument here is very similar to the proof of \cite[Lemma~3]{ON99} and
that of \cite[Lemma~8]{SW12}. We include it here for completeness. Let%
\begin{equation}
g(\alpha,\rho)\equiv\left(  \frac{\alpha-1}{\alpha}\right)  \left(
R-{I}_{\alpha}\left(  A;B\right)  _{\omega}\right)  ,
\end{equation}
where $\omega_{AB}\equiv\rho_{A}^{1/2}\Gamma_{AB}^{\mathcal{N}}\rho_{A}^{1/2}%
$. We can expand $g(\alpha,\rho)$ using
Lemma~\ref{lem:Renyi-mutual-explicit-relation} as
\begin{align}
g(\alpha,\rho) &  =\left(  \frac{\alpha-1}{\alpha}\right)  \left(
R-\frac{\alpha}{\alpha-1}\log\text{Tr}\left\{  \left(  \text{Tr}_{A}\left\{
\rho_{A}^{1-\alpha}\rho_{AB}^{\alpha}\right\}  \right)  ^{\frac{1}{\alpha}%
}\right\}  \right)  \\
&  =\left(  1-\frac{1}{\alpha}\right)  R-\log\text{Tr}\left\{  \left(
\text{Tr}_{A}\left\{  \rho_{A}^{1-\alpha}\rho_{AB}^{\alpha}\right\}  \right)
^{\frac{1}{\alpha}}\right\}  .
\end{align}
Suppose now that $R>I\left(  \mathcal{N}\right)  $, as in the statement of the
lemma. Then for $\forall\rho$ we have that $g\left(  1,\rho\right)  =0$ and by
Lemma~\ref{lem:converge-to-I},%
\begin{equation}
g^{\prime}\left(  1,\rho\right)  =R-I\left(  A;B\right)  _{\omega
}>0.\label{eq:derivative-ge-than-zero}%
\end{equation}
Now, suppose for a contradiction that (\ref{eq:fun-ge-than-zero}) does not
hold, or equivalently, that
\begin{equation}
\forall\beta>1,\,\exists\alpha\in\left(  1,\beta\right)  ,\qquad\min_{\rho
}g(\alpha,\rho)\leq0.\label{eq:fun-le-than-zero}%
\end{equation}
Then there exists a real sequence $\{\alpha_{n}\}$ and a sequence $\{\rho
_{n}\}$ of states in $\mathcal{S}(\mathcal{H}_{A})$ such that
\begin{equation}
\alpha_{n}\in\left(  1,1+\frac{1}{n}\right)  \quad\text{and}\quad g\left(
\alpha_{n},\rho_{n}\right)  \leq0.
\end{equation}
Since $\mathcal{S}(\mathcal{H}_{A})$ is a compact space, there exists a
subsequence of $\rho_{n}$ that converges to some state $\rho_{\infty}%
\in\mathcal{S}(\mathcal{H}_{A})$ as $n\rightarrow\infty$. Relabeling the
subsequence to be $\{\rho_{n}\}$, without loss of generality we can assume
that $\rho_{n}\rightarrow\rho_{\infty}$ as $n\rightarrow\infty$. By the mean
value theorem, we have that
\begin{equation}
\forall n,\,\exists\gamma_{n}\in\left(  1,\alpha_{n}\right)  ,\qquad
g^{\prime}\left(  \gamma_{n},\rho_{n}\right)  =\frac{g\left(  \alpha_{n}%
,\rho_{n}\right)  -g\left(  1,\rho_{n}\right)  }{\alpha_{n}-1}\leq
0.\label{eq:mean-value-thm}%
\end{equation}
Since $g^{\prime}\left(  \alpha,\rho\right)  $ is a continuous function of
$\left(  \alpha,\rho\right)  $, (\ref{eq:mean-value-thm}) yields
\begin{equation}
g^{\prime}\left(  1,\rho_{\infty}\right)  \leq0,
\end{equation}
which contradicts (\ref{eq:derivative-ge-than-zero}).
\end{proof}

Now, Lemma \ref{lem: rate-capacity-gt-zero} and (\ref{eq:new-succ-prob}) yield
our main theorem:

\begin{theorem}
[Strong converse for EA capacity]For any sequence of entanglement-assisted
codes for a channel $\mathcal{N}$ and with rate $R > I\left(  \mathcal{N}%
\right)  $, the success probability decays exponentially to zero as $n
\rightarrow\infty$.
\end{theorem}

\section{Conclusion}

This paper provides an alternate path for establishing a strong converse
theorem for the entanglement-assisted capacity of any quantum channel. The
strong converse theorem, along with the coding theorem from
\cite{PhysRevLett.83.3081,Bennett2002,Hol01a}, refines our understanding of
the entanglement-assisted capacity as a sharp dividing line between what rates
of communication are possible or impossible. The approach taken here is to
exploit the entanglement-assisted \textquotedblleft
meta-converse\textquotedblright\ from \cite{MW12}, several properties of the
sandwiched R\'{e}nyi relative entropy (especially its relation to $\alpha
$-norms) \cite{WWY13,MDSFT13,B13monotone,FL13}, and the main result from
\cite{DJKR06}. The appeal of the present paper is that it demonstrates the
extent to which the powerful Arimoto approach \cite{A73,PV10,SW12} can be
helpful in establishing strong converses, and furthermore, we provide an
operational relevance for the main result in \cite{DJKR06}. The present paper
also adds to the existing body of evidence \cite{WWY13,MDSFT13,MO13,DL13}%
\ that the sandwiched R\'{e}nyi relative entropy is the correct quantum
generalization of the classical concept for all $\alpha>1$. Finally, some of
the ideas in this work might be helpful in solving the open question from
\cite{MW13}\ (i.e., establishing a strong converse theorem for the quantum
capacity of degradable quantum channels).

\textit{Acknowledgements}---We are grateful to Naresh Sharma and Andreas
Winter for insightful discussions and to Mil\'an Mosonyi for some helpful
comments on our paper. We also thank the Department of Physics and Astronomy
and the Center for Computation and Technology at Louisiana State University
for providing startup funds to support this research.

\bibliographystyle{plain}
\bibliography{Ref}

\end{document}